\documentclass[journal]{IEEEtran}

\usepackage{mathrsfs}
\usepackage{mathptmx}
\usepackage{amsmath}
\usepackage{amsfonts}
\usepackage{amssymb}
\usepackage{graphicx}
\usepackage{subfigure}
\usepackage{psfrag}

\newtheorem{thm}{Theorem}

\newtheorem{cor}{Corollary}
\newtheorem{rem}{Remark}

\newtheorem{lem}{Lemma}

\def\Ds{\displaystyle}

\def\tp{\mathrm{T}}

\def\Z{\mathbb{Z}}
\def\R{\mathbb{R}}

\begin{document}

\title{System Equivalence Transformation: Robust Convergence of Iterative Learning Control\\with Nonrepetitive Uncertainties}

\author{Deyuan Meng and Jingyao Zhang 
\thanks{Deyuan Meng and Jingyao Zhang are with the Seventh Research Division, Beihang University (BUAA), Beijing 100191, P. R. China, and also with the School of Automation Science and Electrical Engineering, Beihang University (BUAA), Beijing 100191, P. R. China (e-mail: dymeng@buaa.edu.cn, zhangjingyao@buaa.edu.cn).}
}

\date{}
\maketitle

\begin{abstract}
For iterative learning control (ILC), one of the basic problems left to address is how to solve the contradiction between convergence conditions for the output tracking error and for the input signal (or error). This problem is considered in the current paper, where the robust convergence analysis is achieved for ILC systems in the presence of nonrepetitive uncertainties. A system equivalence transformation (SET) is proposed for ILC such that given any desired reference trajectories, the output tracking problems for general nonsquare multi-input, multi-output (MIMO) systems can be equivalently transformed into those for the specific class of square MIMO systems with the same input and output channels. As a benefit of SET, a unified condition is only needed to guarantee both the uniform boundedness of all system signals and the robust convergence of the output tracking error, which avoids causing the condition contradiction problem in implementing the double-dynamics analysis approach to ILC. Simulation examples are included to demonstrate the validity of our established robust ILC results.
\end{abstract}


\section{Introduction}\label{sec1}

Iterative learning control (ILC) has been proposed as an effective intelligent control approach that has a salient focus on the systems running repetitively over a fixed (time) interval, see, e.g., \cite{ksgha:19} for ventricular assist devices, \cite{hvxh:14} for hard disk drives, \cite{hmhs:19} for flexible structures, and \cite{dcrfd:19} for stimulated lower-limb muscle groups. Until now, ILC has an extensive literature covering linear and nonlinear plants; extensions to plants to consider robustness with respect to external disturbances, initial conditions, and model uncertainties; ILC algorithm design; and applications (for more discussions, see, e.g., the surveys of \cite{bta:06}-\cite{x:11} that provide a good introduction of ILC). Nevertheless, despite the progress that has been made in understanding ILC, more research is required to establish a complete theory. For example, it is generally needed to decide how to implement convergence analysis of ILC. Which should we better resort to, the input (error) or (output) tracking error, and why do the different choices cause condition contradiction in the same ILC convergence problem? In the presence of any desired reference trajectory, how can the tracking ability of ILC be ensured, and further which and how many input channels should we specify for ILC updating to accomplish the tracking tasks? To our best knowledge, there have not been provided any clear answers to these questions in the ILC literature.

There are basically two types of approaches to convergence analysis of ILC in the literature \cite{t:07,mjdy:11b}. The first one is most used in the early stage of ILC, which focuses on exploiting the ILC updating law and resorts to the input convergence to indirectly induce the (output) tracking objective (see, e.g., \cite{lb:96}-\cite{s:05}). Thus, it creates an indirect approach to the ILC convergence analysis, with which convergence properties of all system signals can be accomplished. However, the indirect approach is applicable for only a portion of reference trajectories that are called realizable trajectories and generally needs to make an assumption on the system invertibility. It is usually hard to check the realizability of any desired reference trajectory, which restricts the tracking ability of ILC. By contrast, the second one is a direct approach that contributes directly to gaining the tracking objective based on developing the convergence of ILC tracking error (see, e.g., \cite{zxhh:15}-\cite{j:19b}). Although the direct approach avoids imposing the realizability of desired reference trajectories, it neglects the updating process of input, and as a consequence it is unclear what convergence properties of the input signal are and which and how many input channels are required for the specified tracking task. 

Another fact worth noticing is that the indirect and direct approaches create different convergence conditions in ILC. Take, for example, ILC for multi-input, multi-output (MIMO) linear systems in \cite{acm:07,mjdy:11b}. The indirect and direct approaches for ILC need fundamentally full-column rank and full-row rank of the same Markov parameter (or pulse response coefficient) matrix, respectively, which however causes the condition contradiction in ILC convergence. In circumstances concerning nonrepetitive (or iteration-dependent) uncertainties especially arising in the plant models \cite{mm:14}-\cite{mm:172}, a (indirect and direct) double-dynamics analysis is performed and two different convergence conditions are produced for ILC of MIMO systems. Though the condition contradiction is avoided by harnessing the nonrepetitiveness in ILC, it seems that the ILC convergence results of nonrepetitive systems no longer work for the special case when nonrepetitive uncertainties disappear. This puzzle is because of the condition contradiction that is caused by the indirect and direct analysis approaches to ILC.

This paper contributes to constructing a system equivalence transformation (SET) approach for robust ILC in the presence of nonrepetitive uncertainties, including initial shifts, external disturbances and model uncertainties. It is shown from the SET results that for any desired reference trajectories, the tracking problems for nonsquare MIMO systems can be equivalently transformed into those for the specific class of square MIMO systems whose input and output channels are equal. With this property, the double-dynamics analysis of ILC can be implemented and the condition contradiction in the ILC convergence can be simultaneously avoided. Furthermore, for the tracking of any $p$-channel desired reference trajectory, $p$ input channels are essentially needed to be updated from iteration to iteration and the other input channels (if exist) are unchanged. Another advantage of SET is that a unified condition is sufficient to not only guarantee the uniform boundedness of all system signals but also accomplish the robust convergence of the ILC tracking error. Simulation tests are also implemented to illustrate the effectiveness of our proposed robust ILC method and results.

The remaining of our paper is organized as follows. We give the problem formulation for robust ILC of uncertain systems in Section \ref{sec2}. In Section \ref{sec3}, the SET results are first proposed and then the robust convergence results of MIMO uncertain ILC systems are established, for which a unified condition is only utilized regardless of the nonrepetitive uncertainties. Extensions of the SET-based convergence analysis results are made for ILC systems subject to nonzero relative degrees in Section \ref{sec6}. Simulations and conclusions are provided in Sections \ref{sec4} and \ref{sec5}, respectively.

{\it Notations:} $\left\|A\right\|$ (or $\left\|A\right\|_{\infty}$) is a norm (or the maximum row sum norm) of any matrix $A\in\mathbb{R}^{m\times n}$; $\rho(A)$ is the spectral radius of a square matrix $A\in\mathbb{R}^{n\times n}$; and $\Delta:z_{l}(k)\to \Delta z_{l}(k)\triangleq z_{l+1}(k)-z_{l}(k)$ is a forward difference operator of any vector $z_{l}(k)\in\mathbb{R}^{n}$ with respect to the changing of $l$ for a fixed $k$.

\section{Problem Formulation}\label{sec2}

Let $k\in\mathbb{Z}_{N}\triangleq\left\{0,1,\cdots,N\right\}$ and $l\in\mathbb{Z}_{+}\triangleq\left\{0,1,2,\cdots\right\}$ be the (discrete) time variable and the iteration variable, respectively. Consider a discrete-time, MIMO uncertain system over $k\in\mathbb{Z}_{N}$ and $l\in\mathbb{Z}_{+}$ as
\begin{equation}\label{eq01}
\left\{\aligned
x_{l}(k+1)&=A_{l}(k)x_{l}(k)+B_{l}(k)u_{l}(k)+w_{l}(k)\\
y_{l}(k)&=C_{l}(k)x_{l}(k)+D_{l}(k)u_{l}(k)+v_{l}(k)
\endaligned\right.
\end{equation}

\noindent where $x_{l}(k)\in\mathbb{R}^{n}$, $u_{l}(k)\in\mathbb{R}^{m}$ and $y_{l}(k)\in\mathbb{R}^{p}$ denote the state, input and output, respectively; $A_{l}(k)\in\mathbb{R}^{n\times n}$, $B_{l}(k)\in\mathbb{R}^{n\times m}$, $C_{l}(k)\in\mathbb{R}^{p\times n}$ and $D_{l}(k)\in\mathbb{R}^{p\times m}$ are uncertain system matrices; and $w_{l}(k)\in\mathbb{R}^{n}$ and $v_{l}(k)\in\mathbb{R}^{p}$ are the load and measurement disturbances, respectively. Obviously, the system (\ref{eq01}) is subject to nonrepetitive (or iteration-dependent) uncertainties from not only external disturbances but also plant model matrices.

{\it Tracking objective:} Consider any reference trajectory $r_{l}(k)\in\mathbb{R}^{p}$, and then the objective of ILC is to drive the system (\ref{eq01}) to track the reference trajectory as accurately as possible over $k\in\mathbb{Z}_{N}$ with the increasing of the iteration number $l$, regardless of the nonrepetitive uncertainties. Mathematically, let the tracking error be $e_{l}(k)=r_{l}(k)-y_{l}(k)$, which is ensured to be uniformly bounded such that
\begin{equation}\label{eq02}
\limsup_{l\to\infty}\max_{1\leq k\leq N}\|e_{l}(k)\|\leq\varepsilon
\end{equation}

\noindent for some (relatively small) finite bound $\varepsilon\geq0$. Note that in (\ref{eq02}), the steady error bound $\varepsilon$ depends on all classes of nonrepetitive uncertainties and vanishes when the system (\ref{eq01}) and its tracking task collapse into a repetitive system to track a fixed reference trajectory without nonrepetitive uncertainties.

To reach the robust tracking objective of ILC in the presence of nonrepetitive uncertainties, we apply an ILC algorithm with the following updating law:
\begin{equation}\label{eq03}
u_{l+1}(k)
=u_{l}(k)+\Xi(k)e_{l}(k)
+\Gamma(k)e_{l}(k+1),\quad \forall l\in\mathbb{Z}_{+},k\in\mathbb{Z}_{N}
\end{equation}

\noindent where $\Xi(k)\in\mathbb{R}^{m\times p}$ and $\Gamma(k)\in\mathbb{R}^{m\times p}$ are gain matrices to be selected. Even though the system (\ref{eq01}) is nonrepetitive, the ILC algorithm (\ref{eq03}) employs two repetitive gain matrices to deal with nonrepetitive uncertainties. Furthermore, the selections of gain matrices may adapt to the different relative degree cases of the system (\ref{eq01}).

To carry out the robust analysis of ILC against nonrepetitive uncertainties, we are interested in the nonrepetitive quantities $A_{l}(k)$, $B_{l}(k)$, $C_{l}(k)$, $D_{l}(k)$, $w_{l}(k)$, $v_{l}(k)$, $r_{l}(k)$ and $x_{l}(0)$ required in the ILC tracking of the system (\ref{eq01}) to take the form of
\begin{equation}\label{eq04}
\aligned
A_{l}(k)&=A(k)+\delta_{A}(l,k),
&B_{l}(k)&=B(k)+\delta_{B}(l,k)\\
C_{l}(k)&=C(k)+\delta_{C}(l,k),
&D_{l}(k)&=D(k)+\delta_{D}(l,k)\\
w_{l}(k)&=w(k)+\delta_{w}(l,k),
&v_{l}(k)&=v(k)+\delta_{v}(l,k)\\
r_{l}(k)&=r(k)+\delta_{r}(l,k),
&x_{l}(0)&=x_{0}+\delta_{x_{0}}(l)\\
\endaligned
\end{equation}

\noindent where $A(k)$, $B(k)$, $C(k)$, $D(k)$, $w(k)$, $v(k)$, $r(k)$ and $x_{0}$ represent repetitive (or iteration-independent) quantities of  $A_{l}(k)$, $B_{l}(k)$, $C_{l}(k)$, $D_{l}(k)$, $w_{l}(k)$, $v_{l}(k)$, $r_{l}(k)$ and $x_{l}(0)$, respectively, and $\delta_{A}(l,k)$,  $\delta_{B}(l,k)$, $\delta_{C}(l,k)$, $\delta_{D}(l,k)$, $\delta_{w}(l,k)$, $\delta_{v}(l,k)$, $\delta_{r}(l,k)$ and $\delta_{x_{0}}(l)$ are nonrepetitive uncertainties of them, respectively. For these nonrepetitive uncertainties, we also make a fundamental boundedness assumption.

\begin{enumerate}
\item[(A1)]
{\it Bounded Uncertainties}: For any $k\in\mathbb{Z}_{N}$ and any $l\in\mathbb{Z}_{+}$, the nonrepetitive uncertainties $\delta_{A}(l,k)$, $\delta_{B}(l,k)$, $\delta_{C}(l,k)$, $\delta_{D}(l,k)$, $\delta_{w}(l,k)$, $\delta_{v}(l,k)$, $\delta_{r}(l,k)$ and $\delta_{x_{0}}(l)$ are bounded such that
\begin{equation*}\label{}
\aligned
\left\|\delta_{A}(l,k)\right\|&\leq\overline{\beta}_{A},
&\left\|\delta_{B}(l,k)\right\|&\leq\overline{\beta}_{B}\\
\left\|\delta_{C}(l,k)\right\|&\leq\overline{\beta}_{C},
&\left\|\delta_{D}(l,k)\right\|&\leq\overline{\beta}_{D}\\
\left\|\delta_{w}(l,k)\right\|&\leq\overline{\beta}_{w},
&\left\|\delta_{v}(l,k)\right\|&\leq\overline{\beta}_{v}\\
\left\|\delta_{r}(l,k)\right\|&\leq\overline{\beta}_{r},
&\left\|\delta_{x_{0}}(l)\right\|&\leq\overline{\beta}_{x_{0}}
\endaligned
\end{equation*}

\noindent for some finite bounds $\overline{\beta}_{A}\geq0$, $\overline{\beta}_{B}\geq0$, $\overline{\beta}_{C}\geq0$, $\overline{\beta}_{D}\geq0$, $\overline{\beta}_{w}\geq0$, $\overline{\beta}_{v}\geq0$, $\overline{\beta}_{r}\geq0$ and $\overline{\beta}_{x_{0}}\geq0$.
\end{enumerate}

\begin{rem}\label{rem01}
Because the Assumption (A1) is concerned with only the nonrepetitive uncertainties $\delta_{A}(l,k)$, $\delta_{B}(l,k)$, $\delta_{C}(l,k)$, $\delta_{D}(l,k)$, $\delta_{w}(l,k)$, $\delta_{v}(l,k)$, $\delta_{r}(l,k)$ and $\delta_{x_{0}}(l)$, it is generally reasonable to consider $\overline{\beta}_{A}$, $\overline{\beta}_{B}$, $\overline{\beta}_{C}$, $\overline{\beta}_{D}$, $\overline{\beta}_{w}$, $\overline{\beta}_{v}$, $\overline{\beta}_{r}$ and $\overline{\beta}_{x_{0}}$ to be relatively smaller in comparison with $\max_{0\leq k\leq N}\left\|A(k)\right\|$, $\max_{0\leq k\leq N}\left\|B(k)\right\|$, $\max_{0\leq k\leq N}\left\|C(k)\right\|$, $\max_{0\leq k\leq N}\left\|D(k)\right\|$, $\max_{0\leq k\leq N}\left\|w(k)\right\|$, $\max_{0\leq k\leq N}\left\|v(k)\right\|$, $\max_{0\leq k\leq N}\left\|r(k)\right\|$ and $\left\|x_{0}\right\|$, respectively. Let us also denote
\begin{equation*}\label{}
\aligned
\beta_{A}&=\overline{\beta}_{A}+\max_{0\leq k\leq N}\left\|A(k)\right\|,
&\beta_{B}&=\overline{\beta}_{B}+\max_{0\leq k\leq N}\left\|B(k)\right\|\\
\beta_{C}&=\overline{\beta}_{C}+\max_{0\leq k\leq N}\left\|C(k)\right\|,
&\beta_{D}&=\overline{\beta}_{D}+\max_{0\leq k\leq N}\left\|D(k)\right\|\\
\beta_{w}&=\overline{\beta}_{w}+\max_{0\leq k\leq N}\left\|w(k)\right\|,
&\beta_{v}&=\overline{\beta}_{v}+\max_{0\leq k\leq N}\left\|v(k)\right\|\\
\beta_{r}&=\overline{\beta}_{r}+\max_{0\leq k\leq N}\left\|r(k)\right\|,
&\beta_{x_{0}}&=\overline{\beta}_{x_{0}}+\left\|x_{0}\right\|
\endaligned
\end{equation*}

\noindent which can actually be the bounds for $A_{l}(k)$, $B_{l}(k)$, $C_{l}(k)$, $D_{l}(k)$, $w_{l}(k)$, $v_{l}(k)$, $r_{l}(k)$ and $x_{l}(0)$, respectively, based on (\ref{eq04}).
\end{rem}

\section{SET-Based Convergence Analysis of ILC}\label{sec3}

\subsection{A SET Approach}

To address the robust output tracking problems of the system (\ref{eq01}) with any desired reference outputs subject to nonrepetitive uncertainties, the application of the updating law (\ref{eq03}) basically requires certain assumptions on the system structure associated with the input-output coupling relation that is related closely to the selections of gain matrices (see also \cite{acm:07,mjdy:11b}). We attempt to make use of the gain matrix $\Xi(k)$ in (\ref{eq03}) (i.e., letting $\Gamma(k)\equiv0$) to achieve the tracking objective (\ref{eq02}), for which the following relative degree condition is needed.

\begin{enumerate}
\item[(A2)]
{\it Relative Degree Conditions}: For any $k\in\mathbb{Z}_{N}$, the nominal matrix $D(k)$ has full-row rank.
\end{enumerate}

It is worth highlighting that (A2) resorts to only the nominal matrix $D(k)$, rather than the uncertain system matrix $D_{l}(k)$. In general, if the relative degree condition (A2) can not be satisfied, then the updating law (\ref{eq03}) enables the system (\ref{eq01}) to track not any but only a portion of references that are considered to be {\it realizable references} \cite{mjdy:11b}. However, in contrast, the output tracking of any desired reference needs the relative degree condition (A2) as a basic requirement. From the relative degree condition (A2), we denote $D(k)=\left[D_{1}(k)~D_{2}(k)\right]$ with $D_{1}(k)\in\mathbb{R}^{p\times p}$ and $D_{2}(k)\in\mathbb{R}^{p\times(m-p)}$ such that $D_{1}(k)$ is a nonsingular matrix without loss of generality (otherwise, this can be easily realized based on the elementary transformation). We accordingly denote $\Xi(k)=\left[\Xi_{1}^{\tp}(k)~\Xi_{2}^{\tp}(k)\right]^{\tp}$, where $\Xi_{1}(k)\in\mathbb{R}^{p\times p}$ and $\Xi_{2}(k)\in\mathbb{R}^{(m-p)\times p}$.

Based on the relative degree condition (A2), next we propose a lemma on how to construct a nonsingular transformation matrix for the ILC system (\ref{eq01}) and (\ref{eq03}).

\begin{lem}\label{lem01}
If the following spectral radius condition holds:
\begin{equation}\label{eq05}
\aligned
\rho\left(I-D(k)\Xi(k)\right) <1, \quad \forall  k \in \Z_{N}
\endaligned
\end{equation}

\noindent then there exist four matrices, given by
\begin{equation*}
\aligned
Q_{11}(k)&=D_1(k)\\
Q_{12}(k)&=D_2(k)\\
Q_{21}(k)&=-\Xi_2(k)[D(k)\Xi(k)]^{-1}D_1(k)\\
Q_{22}(k)&=I-\Xi_2(k)[D(k)\Xi(k)]^{-1}D_2(k)
\endaligned
\end{equation*}

\noindent such that a structured block matrix $Q(k)$ in the form of
\begin{equation*}
\aligned
Q(k)=\begin{bmatrix}
       Q_{11}(k) & Q_{12}(k) \\
       Q_{21}(k) & Q_{22}(k)
     \end{bmatrix}
\endaligned
\end{equation*}

\noindent is nonsingular, of which the inverse matrix is structured in the form of
\begin{equation*}
\aligned
Q^{-1}(k)=\begin{bmatrix}
       \widehat{Q}_{11}(k) & \widehat{Q}_{12}(k) \\
       \widehat{Q}_{21}(k) & \widehat{Q}_{22}(k)
     \end{bmatrix}
\endaligned
\end{equation*}

\noindent where
\begin{equation*}
\aligned
\widehat{Q}_{11}(k)&=\Xi_1(k)[D(k)\Xi(k)]^{-1}\\
\widehat{Q}_{12}(k)&=-[D_1(k)]^{-1}D_2(k)\\
\widehat{Q}_{21}(k)&=\Xi_2(k)[D(k)\Xi(k)]^{-1}\\
\widehat{Q}_{22}(k)&=I.
\endaligned
\end{equation*}

\noindent Further, there exists a learning gain matrix $\Xi(k)$ satisfying (\ref{eq05}) if and only if the relative degree condition (A2) holds.
\end{lem}

\begin{proof}
It is straightforward from the matrix theory \cite{hj:85} that there exists $\Xi(k)$ fulfilling (\ref{eq05}) if and only if (A2) holds. Then with (\ref{eq05}), $D(k)\Xi(k)$ is guaranteed to be nonsingular, and as a consequence, the construction of $Q(k)$, together with its inverse $Q^{-1}(k)$, can be easily validated, for which the details are not detailed here.
\end{proof}

Based on Lemma \ref{lem01}, we proceed further to present a theorem to offer a SET result of the ILC system (\ref{eq01}) and (\ref{eq03}). 

\begin{thm}\label{thm01}
For the system (\ref{eq01}) under the updating law (\ref{eq03}) with $\Gamma(k)\equiv0$, $\forall k\in\mathbb{Z}_{N}$, if the spectral radius condition (\ref{eq05}) holds, then there exists a nonsingular linear transformation 
\begin{equation}\label{eq09}
\aligned
u^\ast_l(k)=Q(k)u_l(k)
\triangleq
\begin{bmatrix}
u^\ast_{1,l}(k) \\
u^\ast_{2,l}(k)
\end{bmatrix}
\endaligned
\end{equation}

\noindent where $u^\ast_{1,l}(k) \in \R^p$ and $u^\ast_{2,l}(k) \in \R^{m-p}$, such that
\begin{enumerate}
\item $u^\ast_{2,l}(k)$ is iteration-independent for each time step, i.e.,
\begin{equation}\label{eq10}
\aligned
u^\ast_{2,l}(k)=u^\ast_{2,0}(k),\quad\forall l \in \Z_+, k \in \Z_N
\endaligned
\end{equation}

\item $u^\ast_{l,1}(k)$ is iteratively updated in the form of
\begin{equation}\label{eq11}
\aligned
u^\ast_{1,l+1}(k)=u^\ast_{1,l}(k)+\Xi^\ast(k) e_l(k),\quad\forall l \in \Z_+, k \in \Z_N
\endaligned
\end{equation}

\noindent where the learning gain matrix $\Xi^\ast(k)$ satisfies
\[
\aligned
\Xi^\ast(k)=D(k)\Xi(k)
\endaligned
\]

\item the system (\ref{eq01}) is equivalently transformed into a system driven only by $u^\ast_{l,1}(k)$ over $l \in \Z_+$ and  $k \in \Z_N$, i.e,
\begin{equation}\label{eq12}
\left \{\aligned x_l(k+1)&=A_l(k)x_l(k)+B^\ast_l(k)u^\ast_{1,l}(k)+w^\ast_l(k)\\
y_l(k)&=C_l(k)x_l(k)+D^\ast_l(k)u^\ast_{1,l}(k)+v^\ast_l(k)
\endaligned \right.
\end{equation}

\noindent where
\begin{equation*}
\aligned
B^\ast_l(k)&=B_l(k)\begin{bmatrix}
                      \widehat{Q}_{11}(k) \\
                      \widehat{Q}_{21}(k)
                    \end{bmatrix},
                    \quad
D^\ast_l(k)=D_l(k)\begin{bmatrix}
                      \widehat{Q}_{11}(k) \\
                      \widehat{Q}_{21}(k)
                    \end{bmatrix}\\
w^\ast_l(k)&=w_l(k)\\
&~~~+B_l(k)\begin{bmatrix}
                      \widehat{Q}_{12}(k)Q_{21}(k) & \widehat{Q}_{12}(k)Q_{22}(k)\\
                      \widehat{Q}_{22}(k)Q_{21}(k) & \widehat{Q}_{22}(k)Q_{22}(k)
                    \end{bmatrix}u_0(k)\\
v^\ast_l(k)&=v_l(k)\\
&~~~+D_l(k)\begin{bmatrix}
                      \widehat{Q}_{12}(k)Q_{21}(k) & \widehat{Q}_{12}(k)Q_{22}(k)\\
                      \widehat{Q}_{22}(k)Q_{21}(k) & \widehat{Q}_{22}(k)Q_{22}(k)
                    \end{bmatrix}u_0(k).
\endaligned
\end{equation*}

\noindent Further, $D^\ast_l(k)$ satisfies
\[
D^\ast_l(k)
=I+\delta_D(l,k)
\begin{bmatrix}
\widehat{Q}_{11}(k)\\
\widehat{Q}_{21}(k)
\end{bmatrix}.
\]

\end{enumerate}
\end{thm}

\begin{proof}
Let us employ the construction of $Q(k)$ in Lemma \ref{lem01}, with which we can validate
\begin{equation}\label{eq13}
\aligned
Q(k)\Xi(k)=\begin{bmatrix}
             D(k)\Xi(k) \\
             0
           \end{bmatrix}.
\endaligned
\end{equation}

\noindent If we consider the transformation (\ref{eq09}) for the updating law (\ref{eq03}), then with $\Gamma(k)\equiv0$, $\forall k\in\mathbb{Z}_{N}$, we can derive
\begin{equation*}
\aligned
u^\ast_{l+1}(k)=u^\ast_{l}(k)+[Q(k)\Xi(k)]e_l(k)
\endaligned
\end{equation*}

\noindent which, together with (\ref{eq13}), reads in a structured block form as
\begin{equation}\label{eq14}
\aligned
\begin{bmatrix}
  u^\ast_{1,l+1}(k) \\
  u^\ast_{2,l+1}(k)
\end{bmatrix}=
\begin{bmatrix}
  u^\ast_{1,l}(k) \\
  u^\ast_{2,l}(k)
\end{bmatrix}+
\begin{bmatrix}
  D(k)\Xi(k) \\
  0
\end{bmatrix}e_l(k).
\endaligned
\end{equation}

\noindent Thus, (\ref{eq10}) and  (\ref{eq11}) follow immediately from (\ref{eq14}). By combining (\ref{eq09}) with (\ref{eq10}), we can obtain
\begin{equation*}
\aligned
u^\ast_{2,l}(k)&=\begin{bmatrix}
                    0& I
                  \end{bmatrix}u^\ast_{0}(k)\\
&=\begin{bmatrix}
    Q_{21}(k)&Q_{22}(k)
  \end{bmatrix}u_0(k)
\endaligned
\end{equation*}

\noindent which can be inserted to deduce
\begin{equation}\label{eq15}
\aligned
B_l(k)u_l(k)&=B_l(k)Q^{-1}(k)u^\ast_l(k)\\
&=B_l(k)\begin{bmatrix}
       \widehat{Q}_{11}(k) & \widehat{Q}_{12}(k) \\
       \widehat{Q}_{21}(k) & \widehat{Q}_{22}(k)
     \end{bmatrix}\begin{bmatrix}
                                     u^\ast_{l,1}(k) \\
                                     u^\ast_{l,2}(k)
                                   \end{bmatrix}\\
&=B_l(k)\begin{bmatrix}
       \widehat{Q}_{11}(k)  \\
       \widehat{Q}_{21}(k)
     \end{bmatrix}u^\ast_{1,l}(k)+B_l(k)\begin{bmatrix}
       \widehat{Q}_{12}(k)  \\
       \widehat{Q}_{22}(k)
     \end{bmatrix}u^\ast_{2,l}(k)\\
&=B^\ast_l(k)u^\ast_{1,l}(k)\\
&~~~+B_l(k)\begin{bmatrix}
                      \widehat{Q}_{12}(k)Q_{21}(k) & \widehat{Q}_{12}(k)Q_{22}(k)\\
                      \widehat{Q}_{22}(k)Q_{21}(k) & \widehat{Q}_{22}(k)Q_{22}(k)
                    \end{bmatrix}u_0(k)\\
\endaligned
\end{equation}
and
\begin{equation}\label{eq16}
\aligned
D_l(k)u_l(k)&=D_l(k)Q^{-1}(k)u^\ast_l(k)\\
&=D_l(k)\begin{bmatrix}
       \widehat{Q}_{11}(k) & \widehat{Q}_{12}(k) \\
       \widehat{Q}_{21}(k) & \widehat{Q}_{22}(k)
     \end{bmatrix}\begin{bmatrix}
                                     u^\ast_{l,1}(k) \\
                                     u^\ast_{l,2}(k)
                                   \end{bmatrix}\\
&=D_l(k)\begin{bmatrix}
       \widehat{Q}_{11}(k)  \\
       \widehat{Q}_{21}(k)
     \end{bmatrix}u^\ast_{1,l}(k)+D_l(k)\begin{bmatrix}
       \widehat{Q}_{12}(k)  \\
       \widehat{Q}_{22}(k)
     \end{bmatrix}u^\ast_{2,l}(k)\\
&=D^\ast_l(k)u^\ast_{1,l}(k)\\
&~~~+D_l(k)\begin{bmatrix}
                      \widehat{Q}_{12}(k)Q_{21}(k) & \widehat{Q}_{12}(k)Q_{22}(k)\\
                      \widehat{Q}_{22}(k)Q_{21}(k) & \widehat{Q}_{22}(k)Q_{22}(k)
                    \end{bmatrix}u_0(k).
\endaligned
\end{equation}

\noindent If we substitute (\ref{eq15}) and (\ref{eq16}) into (\ref{eq01}), then we can directly obtain (\ref{eq12}). In addition, $D^\ast_l(k)=I+\delta_D(l,k)[\widehat{Q}^{\tp}_{11}(k),\widehat{Q}^{\tp}_{21}(k)]^{\tp}$ can be validated by taking advantage of the block form of $Q^{-1}(k)$ in Lemma \ref{lem01}.
\end{proof}

\begin{rem}\label{rem03}
For Theorem \ref{thm01}, we perform a nonsingular linear transformation on the updating process with respect to iteration such that two essentially different updating patterns emerge for the input signal: 1) only $p$ channels of the $m$-channel input are needed to be updated for the tracking of any $p$-channel desired output reference, and 2) the remainder $(m-p)$ channels of the input are unchanged during the iteration process. In particular, the input over unchanged channels plays a role as the iteration-independent load disturbances in the ILC process. Further, this brings a benefit that the ILC problems for general (nonsquare) MIMO systems can be equivalently transformed into those for the specific class of square MIMO systems with the same input and output channels. By $D^{\ast}_{l}(k)=I+\delta_D(l,k)[\widehat{Q}^{\tp}_{11}(k),\widehat{Q}^{\tp}_{21}(k)]^{\tp}$, we know from (\ref{eq12}) that in addition to the distinguishment of different input updating patterns, the SET result in Theorem \ref{thm01} can realize the ``one-to-one control'' of the controlled system, especially when the nonrepetitive uncertainty disappears. Another fact worth emphasizing for (\ref{eq12}) is that the disturbance signals $w^{\ast}_{l}(k)$ and $v^{\ast}_{l}(k)$ are independent of its system signals $x_{l}(k)$, $u_{1,l}^{\ast}(k)$ and $y_{l}(k)$.
%
\end{rem}

It is worth noting that the nonsingular transformation matrix in Lemma \ref{lem01} is constructed by taking advantage of the updating pattern of ILC systems. This is actually a salient feature of ILC that resorts to only the relative degree condition (A2)---some basic knowledge of the system structure (see also \cite{acm:07}). In the updating law of input, the learning gain matrix simultaneously plays the role as a filter for the additional driven force caused by the tracking error such that it is possible to perform the SET analysis in Theorem \ref{thm01}. To be more specific, for the ILC system (\ref{eq01}) and (\ref{eq03}) with $\Gamma(k)\equiv0$, $\forall k\in\mathbb{Z}_{N}$, the input dynamics along the iteration axis are described by
\begin{equation}\label{eq17}
\aligned
u_{l+1}(k)&=u_l(k)+\Xi(k)[r_l(k)-y_l(k)]\\
&=u_l(k)+\Xi(k)r_l(k)\\
&~~~-\Xi(k)[C_l(k)x_l(k)+D_l(k)u_l(k)+v_l(k)]\\
&=[I-\Xi(k)D_l(k)]u_l(k)+\zeta_l(k)
\endaligned
\end{equation}

\noindent where $\zeta_l(k)$ is given by
\begin{equation}\label{eq18}
\aligned
\zeta_l(k)=\Xi(k)[r_l(k)-C_l(k)x_l(k)-v_l(k)].
\endaligned
\end{equation}

\noindent Clearly, (\ref{eq17}) represents a discrete system of $u_{l}(k)$ with respect to the iteration $l\in\mathbb{Z}_{+}$ for any time step $k\in\mathbb{Z}_{N}$, where $\zeta_{l}(k)$ denotes the driven force for the updating process of $u_{l}(k)$ from iteration to iteration. It can be seen from (\ref{eq18}) that for $\zeta_{l}(k)$, there exists an iteration-independent ``filter matrix $\Xi(k)$,'' regardless of the presences of any iteration-dependent $r_l(k)$, $C_{l}(k)$, $v_{l}(k)$, and $x_{l}(k)$. Furthermore, it will be seen that the SET result in Theorem \ref{thm01} simplifies convergence conditions of ILC and avoids causing contradiction and conflict in the ILC convergence analysis.

Regarding the convergence analysis of the ILC system (\ref{eq01}) and (\ref{eq03}), we need to consider the nonrepetitive uncertainties. If we take into account the uncertain $D_{l}(k)$, then the robust ILC convergence analysis generally needs a matrix norm condition, instead of (\ref{eq05}), as follows:
\begin{equation}\label{eq06}
\aligned
\|I-D_{l}(k)\Xi(k)\| <1, \quad \forall l\in\Z_{+},k\in\Z_{N}
\endaligned
\end{equation}

\noindent which however can not be easily verified due to the uncertainty $\delta_{D}(l,k)$ of $D_{l}(k)$. But, if we notice the boundedness of $\delta_{D}(l,k)$ in (A1), then we may model $\delta_{D}(l,k)$ in the form of
\begin{equation}\label{eq07}
\aligned
&\delta_{D}(l,k)=E(k)\Sigma_{l}(k)F(k)\\
&\hbox{subject to}~\Sigma_{l}^{\tp}(k)\Sigma_{l}(k)\leq I
\endaligned,\quad\forall l\in\Z_{+},k\in\Z_{N}
\end{equation}

\noindent for some nominal matrices $E(k)$ and $F(k)$ and some uncertain matrix $\Sigma_{l}(k)$, and consequently, we can accomplish (\ref{eq06}) in the sense of the spectral norm provided that the following linear matrix inequality (LMI) condition is satisfied for some scalar $\lambda>0$:
\begin{equation}\label{eq08}
\aligned
\begin{bmatrix}
-I&(\star)&(\star)&(\star)\\
I-D(k)\Xi(k)&-I&(\star)&(\star)\\
0&E^{\tp}(k)&-\lambda I&(\star)\\
F(k)\Xi(k)&0&0&-\lambda I
\end{bmatrix}<0,\quad\forall k\in\Z_{N}
\endaligned
\end{equation}

\noindent where $(\star)$ denotes the terms induced by the symmetry. Clearly, (\ref{eq05}) is only a necessary condition of (\ref{eq08}).

\subsection{Robust ILC Convergence}

Next, we consider the robust tracking problem of the system (\ref{eq01}) under the action of an ILC algorithm given by the updating law (\ref{eq03}). It is worth highlighting that the system (\ref{eq01}) is subject to nonrepetitive model uncertainties arising from $A_{l}(k)$, $B_{l}(k)$, $C_{l}(k)$ and $D_{l}(k)$. According to the robust ILC results of MIMO nonrepetitive uncertain systems (see, e.g., \cite{mm:171}), two conditions are needed to guarantee robust convergence of the ILC system (\ref{eq01}) and (\ref{eq03}), with one condition ensuring the boundedness of all the system signals and the other condition accomplishing the robust convergence of the tracking error. More specifically, let us neglect the nonrepetitive uncertainties, and if we resort to the CM-based approach to the ILC convergence analysis, then for the ILC system (\ref{eq01}) and (\ref{eq03}) with $\Gamma(k)\equiv0$, the boundedness condition collapses into
\begin{equation}\label{eq19}
\aligned
\rho(I-\Xi(k)D(k))<1
,\quad k \in \Z_{N}
\endaligned
\end{equation}

\noindent and the robust ILC convergence condition turns into exactly the spectral radius condition (\ref{eq05}). However, it is straightforward from \cite[Theorem 1.3.20]{hj:85} that (\ref{eq05}) and (\ref{eq19}) can not hold simultaneously for the general nonsquare MIMO systems (i.e., $m\neq p$). What results in the condition contradiction for robust ILC? In fact, this condition contradiction is a basic problem in ILC since the boundedness analysis of system signals resorting to the control input evolution along the iteration axis needs a condition like (\ref{eq19}), while the convergence analysis of ILC resorting to the tracking error requires a condition like (\ref{eq05}) (see, e.g., \cite{mjdy:11b}). Are there any feasible ways to address the fundamental problem of condition contradiction in ILC? To our knowledge, these problems have not been completely solved in ILC.

For the abovementioned problems, if we resort to the SET result in Theorem \ref{thm01}, we can find feasible solutions to them. Namely, for the ILC system (\ref{eq01}) and (\ref{eq03}), if we instead consider (\ref{eq11}) and (\ref{eq12}), then due to
\begin{equation}\label{eq20}
\aligned
I-\Xi^{\ast}(k)D^{\ast}_l(k)&=I-D(k)\Xi(k)\left[I+\delta_D(l,k)\begin{bmatrix}
                                                         \widehat{Q}_{11}(k) \\
                                                         \widehat{Q}_{21}(k)
                                                       \end{bmatrix}\right]\\
&=I-D(k)\Xi(k)-D(k)\Xi(k)\delta_D(l,k)\Xi(k)\\
&~~~\times\left[D(k)\Xi(k)\right]^{-1}\\
&=\left[D(k)\Xi(k)\right]\left[I-D(k)\Xi(k)-\delta_D(l,k)\Xi(k)\right]\\
&~~~\times \left[D(k)\Xi(k)\right]^{-1}\\
&=\left[D(k)\Xi(k)\right]\left[I-D_l(k)\Xi(k)\right]\left[D(k)\Xi(k)\right]^{-1}\\
&=\left[D(k)\Xi(k)\right]\left[I-D^{\ast}_l(k)\Xi^{\ast}(k)\right]\left[D(k)\Xi(k)\right]^{-1}
\endaligned
\end{equation}

\noindent a matrix similarity property emerges between $I-D^{\ast}_l(k)\Xi^{\ast}(k)$ and $I-\Xi^{\ast}(k)D^{\ast}_l(k)$, $\forall k\in\Z_{N}$. In particular, when the nonrepetitive uncertainties are not considered, i.e., $D^{\ast}_l(k)\equiv D^{\ast}(k)$, $\forall l\in\Z_{+}$, (\ref{eq20}) implies
\[
\rho\left(I-D^{\ast}(k)\Xi^{\ast}(k)\right)=\rho\left(I-\Xi^{\ast}(k)D^{\ast}(k)\right),\quad\forall k\in\Z_{N}
\]

\noindent namely, the contradiction problem between (\ref{eq05}) and (\ref{eq19}) can be avoided. As a consequence,  the problem of condition contradiction for the ILC system (\ref{eq01}) and (\ref{eq03}) may naturally disappear for the ILC system (\ref{eq11}) and (\ref{eq12}). By benefiting from this observation, we can establish the following theorem for robust convergence of the ILC system (\ref{eq01}) and (\ref{eq03}).

\begin{thm}\label{thm02}
Consider the system (\ref{eq01}) under the updating law (\ref{eq03}) with $\Gamma(k)\equiv0$, and let Assumption (A1) hold and $\delta_{D}(l,k)$ be in the form of (\ref{eq07}). If the LMI condition (\ref{eq08}) is satisfied, then the boundedness of system signals and the robust tracking task of ILC can be simultaneously achieved, namely,
\begin{enumerate}
\item
$x_l(k)$, $u_l(k)$ and $y_l(k)$ are uniformly bounded, i.e.,
\begin{equation*}
\aligned
\sup_{l \geq 0}\max_{0 \leq k \leq N} \|x_l(k)\| &\leq \beta_x, ~~~\sup_{l \geq 0}\max_{0 \leq k \leq N} \|u_l(k)\| &\leq \beta_u\\
\sup_{l \geq 0}\max_{0 \leq k \leq N} \|y_l(k)\| &\leq \beta_y
\endaligned
\end{equation*}

\noindent for some finite bounds $\beta_x \geq 0$, $\beta_u \geq 0$, and $\beta_y \geq 0$;
  \item $e_l(k)$ is uniformly bounded and fulfills (\ref{eq02}) for some finite bound $\epsilon$ depending continuously on the bounds $\overline{\beta}_A$, $\overline{\beta}_B$, $\overline{\beta}_C$, $\overline{\beta}_D$, $\overline{\beta}_{x_0}$, $\overline{\beta}_w$, $\overline{\beta}_v$, and $\overline{\beta}_r$ of the nonrepetitive uncertainties (i.e., $\epsilon \to 0$ if $\overline{\beta}_A\to 0$, $\overline{\beta}_B\to 0$, $\overline{\beta}_C\to 0$, $\overline{\beta}_D\to 0$, $\overline{\beta}_{x_0}\to 0$, $\overline{\beta}_w\to 0$, $\overline{\beta}_v\to 0$, and $\overline{\beta}_r\to 0$).
\end{enumerate}
\end{thm}

\begin{proof}
{\it 1):} Note that (\ref{eq05}) holds under the LMI condition (\ref{eq08}). For the ILC system (\ref{eq01}) and (\ref{eq03}), we can leverage the SET result in Theorem \ref{thm01}. With (\ref{eq11}) and (\ref{eq12}), we can obtain
\begin{equation}\label{eq21}
\aligned
u^\ast_{1,l+1}(k)&=u^\ast_{1,l}(k)+\Xi^\ast(k)[r_l(k)-y_l(k)]\\
&=u^\ast_{1,l}(k)+\Xi^\ast(k)r_l(k)\\
&~~~-\Xi^\ast(k)[C_l(k)x_l(k)+D^\ast_l(k)u^\ast_{1,l}(k)+v^\ast_l(k)]\\
&=[I-\Xi^\ast(k)D^\ast_l(k)]u^\ast_{1,l}(k)+\zeta^\ast_l(k)
\endaligned
\end{equation}

\noindent where $\zeta^\ast_l(k)$ is given by
\[
\aligned
\zeta^\ast_l(k)=\Xi^\ast(k)[r_l(k)-C_l(k)x_l(k)-v^\ast_l(k)].
\endaligned
\]

\noindent Then inserting (\ref{eq20}) into (\ref{eq21}) yields
\[
\aligned
u^\ast_{1,l+1}(k)&=\left[D(k)\Xi(k)\right]\left[I-D_l(k)\Xi(k)\right]\left[D(k)\Xi(k)\right]^{-1}u^\ast_{1,l}(k)\\
&~~~+\zeta^\ast_l(k)
\endaligned
\]

\noindent which, together with $\overline{u}^\ast_{1,l}(k)=\left[D(k)\Xi(k)\right]^{-1}u^\ast_{1,l}(k)$, leads to
\begin{equation}\label{eq22}
\aligned
\overline{u}^\ast_{1,l+1}(k)&=\left[I-D_l(k)\Xi(k)\right]\overline{u}^\ast_{1,l}(k)+\overline{\zeta}^\ast_l(k)
\endaligned
\end{equation}

\noindent where $\overline{\zeta}^\ast_l(k)$ fulfills
\[
\aligned
\overline{\zeta}^\ast_l(k)&=\left[D(k)\Xi(k)\right]^{-1}\zeta^\ast_l(k)\\
&=\left[D(k)\Xi(k)\right]^{-1}\Xi^\ast(k)[r_l(k)-C_l(k)x_l(k)-v^\ast_l(k)]\\
&=r_l(k)-C_l(k)x_l(k)-v_l(k)\\
&~~~-D_l(k)\begin{bmatrix}
                      \widehat{Q}_{12}(k)Q_{21}(k) & \widehat{Q}_{12}(k)Q_{22}(k)\\
                      \widehat{Q}_{22}(k)Q_{21}(k) & \widehat{Q}_{22}(k)Q_{22}(k)
                    \end{bmatrix}u_0(k).\\
\endaligned
\]

\noindent Besides, we can redescribe (\ref{eq12}) in the form of
\begin{equation}\label{eq23}
\aligned
x_l(k+1)&=A_l(k)x_l(k)+B^\ast_l(k)u^\ast_{1,l}(k)+w^\ast_l(k)\\
&=A_l(k)x_l(k)+B^\ast_l(k) \left[D(k)\Xi(k)\right] \left[D(k)\Xi(k)\right]^{-1}u^\ast_{1,l}(k)\\
&~~~+w^\ast_l(k)\\
&=A_l(k)x_l(k)+B^\ast_l(k) D(k)\Xi(k) \overline{u}^\ast_{1,l}(k)+w^\ast_l(k)
\endaligned
\end{equation}

\noindent for which $w^\ast_l(k)$ is uniformly bounded for any $k \in \Z_N$ and any $l \in \Z_+$, i.e.,
\begin{equation*}
\aligned
\|w^\ast_l(k) \|& \leq \beta_w+ \beta_B \max_{0 \leq k \leq N} \left\| \begin{bmatrix}
                                                                      \widehat{Q}_{12}(k) \\
                                                                      \widehat{Q}_{22}(k)
                                                                    \end{bmatrix} \right\|\\
&~~~~~~~~~\times \max_{0 \leq k \leq N} \left\|\begin{bmatrix}
                                           Q_{21}(k) & Q_{22}(k)
                                         \end{bmatrix} \right\| \max_{0 \leq k \leq N} \|u_0(k) \|.
\endaligned
\end{equation*}

\noindent Under the condition (\ref{eq08}), (\ref{eq06}) can be achieved in the sense of the spectral norm. Based on (\ref{eq22}) and (\ref{eq23}) and with \cite[Lemma 2]{mm:171}, we implement the inductive analysis for $k \in \Z_N$ by following the same lines as the proof of \cite[Theorem 1]{mm:171} to obtain that $x_l(k)$ and $\overline{u}^\ast_{1,l}(k)$ are uniformly bounded such that
\begin{equation}\label{eq24}
\aligned
\sup_{l \geq 0}\max_{0 \leq k \leq N} \|x_l(k)\| \leq \beta_x, ~~\sup_{l \geq 0}\max_{0 \leq k \leq N} \|\overline{u}^\ast_{1,l}(k)\| \leq \beta_{\overline{u}^\ast_1}
\endaligned
\end{equation}

\noindent for some finite bounds $\beta_x \geq 0$ and $\beta_{\overline{u}^\ast_1} \geq 0$. In addition, noticing $\overline{u}^\ast_{1,l}(k)=\left[D(k)\Xi(k)\right]^{-1}u^\ast_{1,l}(k)$, we can directly derive that
\begin{equation}\label{eq25}
\aligned
 \|u^\ast_{1,l}(k)\| &\leq \left \|D(k)\Xi(k) \right\| \|\overline{u}^\ast_{1,l}(k)\|\\
 &\leq \max_{0 \leq k \leq N} \left \|D(k)\Xi(k) \right\| \beta_{\overline{u}^\ast_1}\\
 & \triangleq \beta_{u^\ast_1},~~ \forall l \in \Z_+, k \in \Z_N.
\endaligned
\end{equation}
Also, we use (\ref{eq10}) to obtain
\[
\aligned
\|u^\ast_{2,l}(k) \|&\equiv \| u^\ast_{2,0}(k)\|\\
& \leq \|\begin{bmatrix}
           Q_{21}(k) & Q_{21}(k)
         \end{bmatrix} \| \|u_0(k) \|\\
& \leq \max_{k \in \Z_N} \{\|\begin{bmatrix}
           Q_{21}(k) & Q_{21}(k)
         \end{bmatrix} \| \|u_0(k) \|\}\\
& \triangleq \beta_{u^\ast_2},~~ \forall l \in \Z_+, k \in \Z_N
\endaligned
\]

\noindent which, together with (\ref{eq25}), further yields
\begin{equation*}
\aligned
\sup_{l \geq 0}\max_{0 \leq k \leq N} \|u^\ast_l(k)\| &\leq \sup_{l \geq 0}\max_{0 \leq k \leq N} \|u^\ast_{1,l}(k)\|\\
&~~~+\sup_{l \geq 0}\max_{0 \leq k \leq N} \|u^\ast_{2,l}(k)\|\\
&\leq \beta_{u^\ast}
\endaligned
\end{equation*}

\noindent where $\beta_{u^\ast}=\beta_{u^\ast_1}+\beta_{u^\ast_2}$. It thus follows
\begin{equation*}
\aligned
\sup_{l \geq 0}\max_{0 \leq k \leq N} \|u_l(k)\|\leq &\beta_{u^\ast} \max_{k \in \Z_N} \| Q^{-1}(k)\| \triangleq \beta_u.
\endaligned
\end{equation*}

\noindent Then, we can validate $\sup_{l \geq 0}\max_{0 \leq k \leq N} \|y_l(k)\| \leq \beta_y$ holds for $\beta_y=\beta_C \beta_x +\beta_D \beta_u  +\beta_v$. The uniform boundedness is achieved for all system signals.

{\it 2):} By using (\ref{eq01}) and (\ref{eq03}), we study $\Delta x_l(k)$ and can obtain
\begin{equation}\label{eq26}
\aligned
\Delta x_l(k+1)&=x_{l+1}(k+1)-x_l(k+1)\\
&= A_{l+1}(k)x_{l+1}(k)+B_{l+1}(k)u_{l+1}(k)+w_{l+1}(k)\\
&~~~-A_{l}(k)x_{l}(k)-B_{l}(k)u_{l}(k)-w_{l}(k)\\
&=A_l(k)\Delta x_l(k) + x_{l+1}(k) \Delta A_l(k)\\
&~~~+B_l(k)\Delta u_l(k) + u_{l+1}(k) \Delta B_l(k) +\Delta w_{l}(k)\\
&=B_l(k)\Xi(k) e_l(k)+ A_l(k)\Delta x_l(k) + x_{l+1}(k) \Delta A_l(k)\\
&~~~+u_{l+1}(k) \Delta B_l(k) +\Delta w_{l}(k).
\endaligned
\end{equation}
If we exploit (\ref{eq26}), then the tracking error $e_l(k)$ satisfies
\begin{equation}\label{eq27}
\aligned
e_{l+1}(k)&=r_{l+1}(k)-y_{l+1}(k)\\
&=r_l(k)-y_{l}(k)+y_{l}(k)-y_{l+1}(k)+\Delta r_l(k)\\
&=e_l(k)+\Delta r_l(k)+C_l(k)x_l(k)+D_l(k)u_l(k)+v_l(k)\\
&~~~-C_{l+1}(k)x_{l+1}(k)-D_{l+1}(k)u_{l+1}(k)-v_{l+1}(k)\\
&=e_l(k)+\Delta r_l(k)-C_l(k) \Delta x_l(k)-x_{l+1}(k) \Delta C_l(k)\\
&~~~-D_l(k) \Delta u_l(k)-u_{l+1}(k) \Delta D_l(k)-\Delta v_l(k)\\
&=[I-D_l(k)\Xi(k)] e_l(k) +\tau_l(k)
\endaligned
\end{equation}

\noindent where $\tau_l(k)$ is given by
\[
\aligned
\tau_l(k)&=-C_l(k) \Delta x_l(k)-x_{l+1}(k) \Delta C_l(k)-u_{l+1}(k) \Delta D_l(k)\\
&~~~+\Delta r_l(k)-\Delta v_l(k).
\endaligned
\]

\noindent Based on (\ref{eq26}) and (\ref{eq27}) and with \cite[Lemma 1]{mm:171}, we can follow the same steps as the proof of \cite[Theorem 3]{mm:171} to develop that under the condition (\ref{eq08}), not only is $e_l(k)$ uniformly bounded but also (\ref{eq02}) is ensured.
\end{proof}

\begin{rem}\label{rem04}
In Theorem \ref{thm02}, the boundedness and convergence analyses of ILC systems are exploited with a unified condition, regardless of the considered nonsquare plants and the adopted double-dynamics analysis approach to ILC. This benefits from the SET-based result in Theorem \ref{thm01} that is helpful to deal with the common condition contradiction problem arising from the double-dynamics analysis approach to ILC (see, e.g., \cite{mm:171}). Moreover, it is worth highlighting that the results of Theorems \ref{thm01} and \ref{thm02} are robust with respect to nonrepetitive uncertainties, including not only initial state shifts and external disturbances but also reference trajectories and model uncertainties in ILC. This in fact contributes to enhancing the robust ILC framework of nonrepetitive uncertain systems (see, e.g., \cite{mm:171,mm:172}).
\end{rem}

\section{Extensions to ILC Systems with Nonzero Relative Degrees}\label{sec6}

\subsection{SET-Based Convergence Analysis}

If the nonrepetitive uncertainties of $B_l(k)$ and $C_l(k)$ do not exist and the control input can not directly influence the output, i.e., $B_l(k)\equiv B(k)$, $C_l(k)\equiv C(k)$ and $D_l(k)\equiv0$ are satisfied, then the system (\ref{eq01}) has a nonzero relative degree, which can be simplified as
\begin{equation}\label{eq28}
\left\{\aligned
x_{l}(k+1)&=A_{l}(k)x_{l}(k)+B(k)u_{l}(k)+w_{l}(k)\\
y_{l}(k)&=C(k)x_{l}(k)+v_{l}(k).
\endaligned\right.
\end{equation}

\noindent For the system \eqref{eq28}, we can choose to adopt the gain matrix $\Gamma(k)$ in (\ref{eq03}) to accomplish the tracking objective (\ref{eq02}) (i.e., setting $\Xi(k)\equiv0$, $\forall k\in\Z_{N}$). For this case, the corresponding relative degree condition is proposed as follows.

\begin{enumerate}
\item [(A3)]
{\it Relative Degree Condition}: For any $k\in\mathbb{Z}_{N-1}$, the coupling matrix $C(k+1)B(k)$ is of full-row rank.
\end{enumerate}

From the relative degree condition (A3), we denote $B(k)=\left[B_{1}(k)~B_{2}(k)\right]$ with $B_{1}(k)\in\mathbb{R}^{p\times p}$ and $B_{2}(k)\in\mathbb{R}^{p\times(m-p)}$ such that $C(k+1)B_{1}(k)$ is a nonsingular matrix without any loss of generality (otherwise, this can also be realized based on the elementary transformation). We accordingly denote $\Gamma(k)=\left[\Gamma_{1}^{\tp}(k)~\Gamma_{2}^{\tp}(k)\right]^{\tp}$, where $\Gamma_{1}(k)\in\mathbb{R}^{p\times p}$ and $\Gamma_{2}(k)\in\mathbb{R}^{(m-p)\times p}$.

Similar to Lemma \ref{lem01}, we can provide a lemma for the ILC system (\ref{eq28}) and (\ref{eq03}) with $\Xi(k)\equiv0$ to construct a nonsingular transformation matrix under the relative degree condition (A3).

\begin{lem}\label{lem02}
If
\begin{equation}\label{eq29}
\rho\left(I-C(k+1)B(k)\Gamma(k)\right)<1,\quad\forall k\in\mathbb{Z}_{N-1}
\end{equation}

\noindent then a structured matrix $P(k)$ can be constructed as
\[P(k)
=\begin{bmatrix}P_{11}(k)&P_{12}(k)\\
P_{21}(k)&P_{22}(k)\end{bmatrix}\]

\noindent where
\[\aligned
P_{11}(k)&=C(k+1)B_{1}(k)\\
P_{12}(k)&=C(k+1)B_{2}(k)\\
P_{21}(k)&=-\Gamma_{2}(k)\left(C(k+1)B(k)\Gamma(k)\right)^{-1}C(k+1)B_{1}(k)\\
P_{22}(k)&=I-\Gamma_{2}(k)\left(C(k+1)B(k)\Gamma(k)\right)^{-1}C(k+1)B_{2}(k).
\endaligned\]

\noindent Further, $P(k)$ is a nonsingular matrix whose inverse matrix is given as
\[P^{-1}(k)
=\begin{bmatrix}\widehat{P}_{11}(k)&\widehat{P}_{12}(k)\\
\widehat{P}_{21}(k)&\widehat{P}_{22}(k)\end{bmatrix}\]

\noindent where
\[\aligned
\widehat{P}_{11}(k)&=\Gamma_1(k)\left(C(k+1)B(k)\Gamma(k)\right)^{-1}\\
\widehat{P}_{12}(k)&=-\left(C(k+1)B_{1}(k)\right)^{-1}C(k+1)B_{2}(k)\\
\widehat{P}_{21}(k)&=\Gamma_{2}(k)\left(C(k+1)B(k)\Gamma(k)\right)^{-1}\\
\widehat{P}_{22}(k)&=I.
\endaligned\]

\noindent In addition, there exists a learning gain matrix $\Gamma(k)$ satisfying (\ref{eq29}) if and only if the relative degree condition (A3) holds.
\end{lem}

\begin{proof}
The proof of Lemma \ref{lem02} follows the same steps as that of Lemma \ref{lem01}, which is omitted here.
\end{proof}

With the nonsingular transformation matrix $P(k)$ in Lemma \ref{lem02}, a SET analysis of the ILC system (\ref{eq28}) and (\ref{eq03}) with $\Xi(k)\equiv0$ can be implemented, which is stated in the following theorem.

\begin{thm}\label{thm03}
For the system (\ref{eq28}) under the updating law (\ref{eq03}) with $\Xi(k)\equiv0$, if the spectral radius condition (\ref{eq29}) holds, then there exists a nonsingular linear transformation
\begin{equation}\label{eq30}
u_{l}^{\ast}(k)=P(k)u_{l}(k)\triangleq\begin{bmatrix}u_{1,l}^{\ast}(k)\\u_{2,l}^{\ast}(k)\end{bmatrix}
\end{equation}
\noindent where $u_{1,l}^{\ast}(k)\in\mathbb{R}^{p}$ and $u_{2,l}^{\ast}(k)\in\mathbb{R}^{m-p}$, such that
\begin{enumerate}
\item
$u_{2,l}^{\ast}(k)$ is iteration-independent for each time step, i.e.,
\begin{equation}\label{eq31}
\aligned
u_{2,l}^{\ast}(k)&\equiv u_{2,0}^{\ast}(k),\quad\forall l\in\mathbb{Z}_{+},k\in\mathbb{Z}_{N-1}
\endaligned
\end{equation}

\item
$u_{1,l}^{\ast}(k)$ is iteratively updated in the form of
\begin{equation}\label{eq32}
u_{1,l+1}^{\ast}(k)=u_{1,l}^{\ast}(k)+\Gamma^{\ast}(k)e_{l}(k+1),\quad\forall l\in\mathbb{Z}_{+},k\in\mathbb{Z}_{N-1}
\end{equation}

\noindent where the learning gain matrix $\Gamma^{\ast}(k)$ satisfies
\[\Gamma^{\ast}(k)=C(k+1)B(k)\Gamma(k)\]

\item
the system (\ref{eq01}) is equivalently transformed into a system driven only by $u_{1,l}^{\ast}(k)$ over $k\in\mathbb{Z}_{N}$ and $l\in\mathbb{Z}_{+}$, i.e.,
\begin{equation}\label{eq33}
\left\{\aligned
x_{l}(k+1)&=A_{l}(k)x_{l}(k)+B^{\ast}(k)u_{1,l}^{\ast}(k)+w_{l}^{\ast}(k)\\
y_{l}(k)&=C(k)x_{l}(k)+v_{l}(k)
\endaligned\right.
\end{equation}

\noindent where
\[\aligned
B^{\ast}(k)&=B_{1}(k)\widehat{P}_{11}(k)+B_{2}(k)\widehat{P}_{21}(k)\\
w_{l}^{\ast}(k)&=w_{l}(k)\\
&~~~+B(k)\begin{bmatrix}\widehat{P}_{12}(k)P_{21}(k)&\widehat{P}_{12}(k)P_{22}(k)\\
\widehat{P}_{22}(k)P_{21}(k)&\widehat{P}_{22}(k)P_{22}(k)\end{bmatrix}u_{0}(k).
\endaligned\]

\noindent Further, $C(k+1)B^{\ast}(k)=I$ holds for (\ref{eq28}) over $k\in\mathbb{Z}_{N-1}$.
\end{enumerate}
\end{thm}

\begin{proof}
We can directly follow the same lines as the proof of Theorem \ref{thm01} to establish the results of Theorem \ref{thm03}, for which the proof details are omitted here for simplicity.
\end{proof}

\begin{rem}\label{rem05}
In Theorem \ref{thm03}, the decoupling control is realized for the input and the output of the square MIMO system (\ref{eq33}). We can actually validate
\begin{equation*}\label{}
\aligned
y_{l}(k+1)
&=C(k+1)B^{\ast}(k)u_{1,l}^{\ast}(k)+C(k+1)A_{l}(k)x_{l}(k)\\
&~~~+C(k+1)w_{l}^{\ast}(k)+v_{l}(k+1)
\endaligned
\end{equation*}

\noindent which, together with $C(k+1)B^{\ast}(k)=I$, immediately leads to
\[
\aligned
y_{l}(k+1)
&=u_{1,l}^{\ast}(k)+C(k+1)A_{l}(k)x_{l}(k)+C(k+1)w_{l}^{\ast}(k)\\
&~~~+v_{l}(k+1).
\endaligned
\]

\noindent This clearly shows that the SET approach can gain the ``one-to-one control'' of the considered ILC systems, in spite of the nonzero relative degrees.
%
\end{rem}

The condition contradiction also occurs for the ILC system (\ref{eq28}) and (\ref{eq03}). Specifically, if we employ \cite[Theorems 1 and 3]{mm:171}, then for the ILC system (\ref{eq28}) and (\ref{eq03}) with $\Xi(k)\equiv0$, the boundedness condition of \cite[Condition ($\mathcal{C}_{1}$)]{mm:171} turns into
\begin{equation}\label{eq34}
\rho \left(I-\Gamma(k)C(k+1)B(k)\right)<1,\quad\forall k\in\mathbb{Z}_{N-1}
\end{equation}

\noindent and the robust ILC convergence condition of \cite[Condition ($\mathcal{C}_{2}$)]{mm:171} becomes exactly (\ref{eq29}). However, (\ref{eq29}) and (\ref{eq34}) generally contradict with each other. In contrast with this, if we consider (\ref{eq32}) and (\ref{eq33}) instead of (\ref{eq28}) and (\ref{eq03}), then due to the fact that
\begin{equation}\label{eq35}
\aligned
I-\Gamma^{\ast}(k)C(k+1)B^{\ast}(k)
&=I-C(k+1)B^{\ast}(k)\Gamma^{\ast}(k)\\
&=I-C(k+1)B(k)\Gamma(k),\quad\forall k\in\mathbb{Z}_{N-1}
\endaligned
\end{equation}

\noindent the condition contradiction for (\ref{eq28}) and (\ref{eq03}) vanishes for (\ref{eq32}) and (\ref{eq33}). By taking advantage of this property, we can further establish the following theorem for robust convergence of the ILC system (\ref{eq28}) and (\ref{eq03}).

\begin{thm}\label{thm04}
Consider the system (\ref{eq28}) under the updating law (\ref{eq03}) with $\Xi(k)\equiv0$, and let Assumption (A1) hold. If the spectral radius condition (\ref{eq29}) is satisfied, then the boundedness of system signals and the robust tracking task of ILC can be simultaneously achieved, namely,
\begin{enumerate}
\item
$x_{l}(k)$, $u_{l}(k)$ and $y_{l}(k)$ are uniformly bounded, i.e.,
\[\aligned
\sup_{l\geq0}\max_{0\leq k\leq N}\|x_{l}(k)\|&\leq\beta_{x},
&\sup_{l\geq0}\max_{0\leq k\leq N-1}\|u_{l}(k)\|&\leq\beta_{u}\\
\sup_{l\geq0}\max_{0\leq k\leq N}\|y_{l}(k)\|&\leq\beta_{y}
\endaligned\]

\noindent for some finite bounds $\beta_{x}\geq0$, $\beta_{u}\geq0$ and $\beta_{y}\geq0$;

\item
$e_{l}(k)$ is uniformly bounded and fulfills (\ref{eq02}) for some finite bound $\varepsilon$ depending continuously on the bounds $\overline{\beta}_{A}$, $\overline{\beta}_{x_{0}}$, $\overline{\beta}_{w}$, $\overline{\beta}_{v}$ and $\overline{\beta}_{r}$ of the nonrepetitive uncertainties (i.e., $\varepsilon\to0$ if $\overline{\beta}_{A}\to0$, $\overline{\beta}_{x_{0}}\to0$, $\overline{\beta}_{w}\to0$, $\overline{\beta}_{v}\to0$ and $\overline{\beta}_{r}\to0$).
\end{enumerate}
\end{thm}

\begin{proof}
Note that the spectral radius condition (\ref{eq29}) offers a necessary and sufficient guarantee to determine some matrix norm to ensure
\[
\left\|I-C(k+1)B(k)\Gamma(k)\right\|<1,\quad\forall k\in\Z_{N-1}.
\]

\noindent Then by considering Lemma \ref{lem02} and Theorem \ref{thm03}, we can prove Theorem \ref{thm04} by following the same steps as the proof of Theorem \ref{thm02}, for which the details are omitted here.
\end{proof}

With Theorem \ref{thm04}, we can see that despite nonzero system relative degrees, the boundedness and convergence analyses can be developed for ILC systems through a unified condition. This, together with Theorem \ref{thm03}, enhances the benefit of the SET-based results in implementing the ILC analysis. 
%

\subsection{Discussions on Linear ILC}

Of particular note is the application of Theorems \ref{thm03} and \ref{thm04} to conventional ILC systems without nonrepetitive uncertainties. More specifically, if $\delta_{A}(l,k)$, $\delta_{x_{0}}(l)$, $\delta_{w}(l,k)$ and $\delta_{v}(l,k)$ do not exist in (\ref{eq28}) and (\ref{eq04}), then it collapses into a repetitive MIMO system over $k\in\mathbb{Z}_{N}$ and $l\in\mathbb{Z}_{+}$ as
\begin{equation}\label{eq36}
\left\{\aligned
x_{l}(k+1)&=A(k)x_{l}(k)+B(k)u_{l}(k)+w(k)\\
y_{l}(k)&=C(k)x_{l}(k)+v(k)\\
x_{l}(0)&=x_{0}
\endaligned\right.
\end{equation}

\noindent and if $\delta_{r}(l,k)$ disappears, then instead of (\ref{eq02}), a perfect tracking task is of interest as
\begin{equation}\label{eq37}
\lim_{l\to\infty}\left[r(k)-y_{l}(k)\right]=0,\quad\forall k=1,2,\cdots,N.
\end{equation}

\noindent Correspondingly, the updating law (\ref{eq03}) with $\Xi(k)=0$ becomes
\begin{equation}\label{eq38}
u_{l+1}(k)=u_{l}(k)+\Gamma(k)\left[r(k+1)-y_{l}(k+1)\right],\forall l\in\mathbb{Z}_{+},k\in\mathbb{Z}_{N-1}.
\end{equation}

Based on Theorems \ref{thm03} and \ref{thm04}, the perfect ILC tracking result can be established for (\ref{eq36}) and (\ref{eq38}) in the following corollary.

\begin{cor}\label{cor01}
Consider the system (\ref{eq36}), and let the updating law (\ref{eq38}) be applied. The perfect tracking task (\ref{eq37}) is achieved, together with the uniform boundedness of all system signals being ensured, if and only if the spectral radius condition (\ref{eq29}) is satisfied.
\end{cor}

\begin{proof}
Based on the SET result in Theorem \ref{thm03} and stability results of linear systems (see, e.g., \cite[Theorem 22.11]{r:96}), this corollary can be developed by following the same way as the proof of Theorem \ref{thm04}.
\end{proof}

To develop the result of Corollary \ref{cor01}, either direct or indirect analysis approach to ILC is employed by resorting to either the tracking error or the input for the ILC convergence analysis in the literature (see, e.g., \cite{mjdy:11b}). However, Corollary \ref{cor01} can involve the double-dynamics analysis processes by taking advantage of the SET result in Theorem \ref{thm03}. Further, for the tracking of any specified reference trajectory $r(k)\in\mathbb{R}^{p}$ of interest over $k\in\mathbb{Z}_{N}$, the ILC systems are needed to possess at least $p$ input channels, and the corresponding input and state trajectories can also be determined.
\begin{itemize}
\item
Given any initial input $u_{0}(k)$ over $k\in\mathbb{Z}_{N-1}$, $\lim_{l\to\infty}u_{l}(k)$ for $k\in\mathbb{Z}_{N-1}$ and $\lim_{l\to\infty}x_{l}(k)$ for $k\in\mathbb{Z}_{N}$ both exist and depend heavily on $u_{0}(k)$. We can follow Theorem \ref{thm03} to see that we can determine a unique $u_{1,\infty}^{\ast}(k)=\lim_{l\to\infty}u_{1,l}^{\ast}(k)$, $\forall k\in\mathbb{Z}_{N-1}$ to produce the specified reference $r(k)$. This, together with (\ref{eq31}), yields that $u_{\infty}(k)=\lim_{l\to\infty}u_{l}(k)$ takes the form of
\[\aligned
u_{\infty}(k)
&=P^{-1}(k)u_{\infty}^{\ast}(k)
=P^{-1}(k)\begin{bmatrix}u_{1,\infty}^{\ast}(k)\\u_{2,0}^{\ast}(k)\end{bmatrix}\\
&=P^{-1}(k)\begin{bmatrix}u_{1,\infty}^{\ast}(k)\\\left[P_{21}(k)~P_{22}(k)\right]u_{0}(k)\end{bmatrix},\quad\forall k\in\mathbb{Z}_{N-1}.
\endaligned
\]

\noindent Thus, if $u_{0}(k)$ is specified, the unique control input learnt via ILC can be accordingly determined. A consequence of this is that the state convergence $x_{\infty}(k)=\lim_{l\to\infty}x_{l}(k)$, $\forall k\in\mathbb{Z}_{N}$ can be decided by (\ref{eq36}).
\end{itemize}

The above discussions indicate that the SET analysis helps to disclose the learning nature of ILC for the output tracking of any specified reference trajectories. This can particularly enrich the fundamental convergence analysis results for linear ILC (see, e.g., \cite{bta:06,acm:07}).

\begin{figure}[!t]
  \centering
  \includegraphics[width=0.8\hsize]{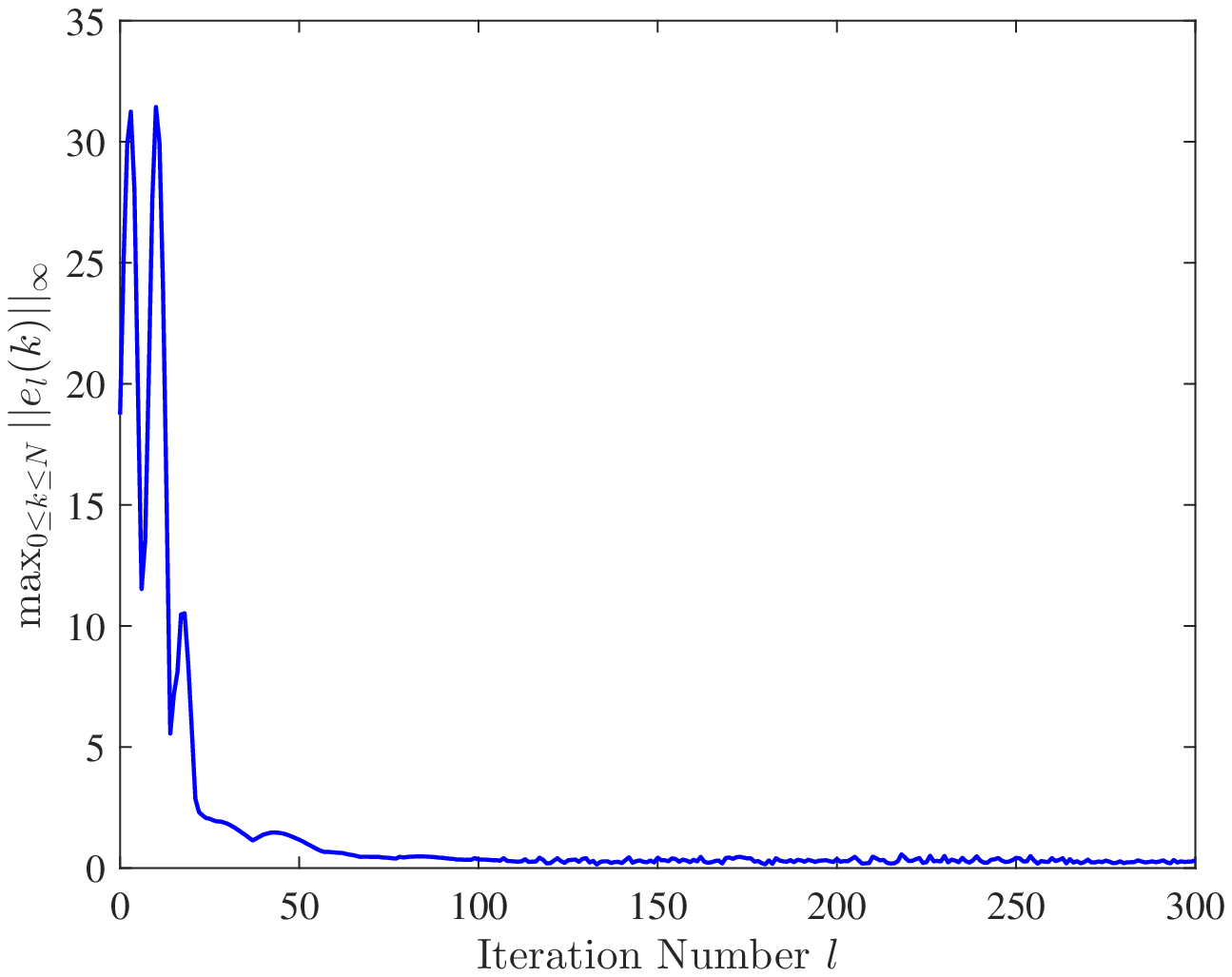}
  \includegraphics[width=0.8\hsize]{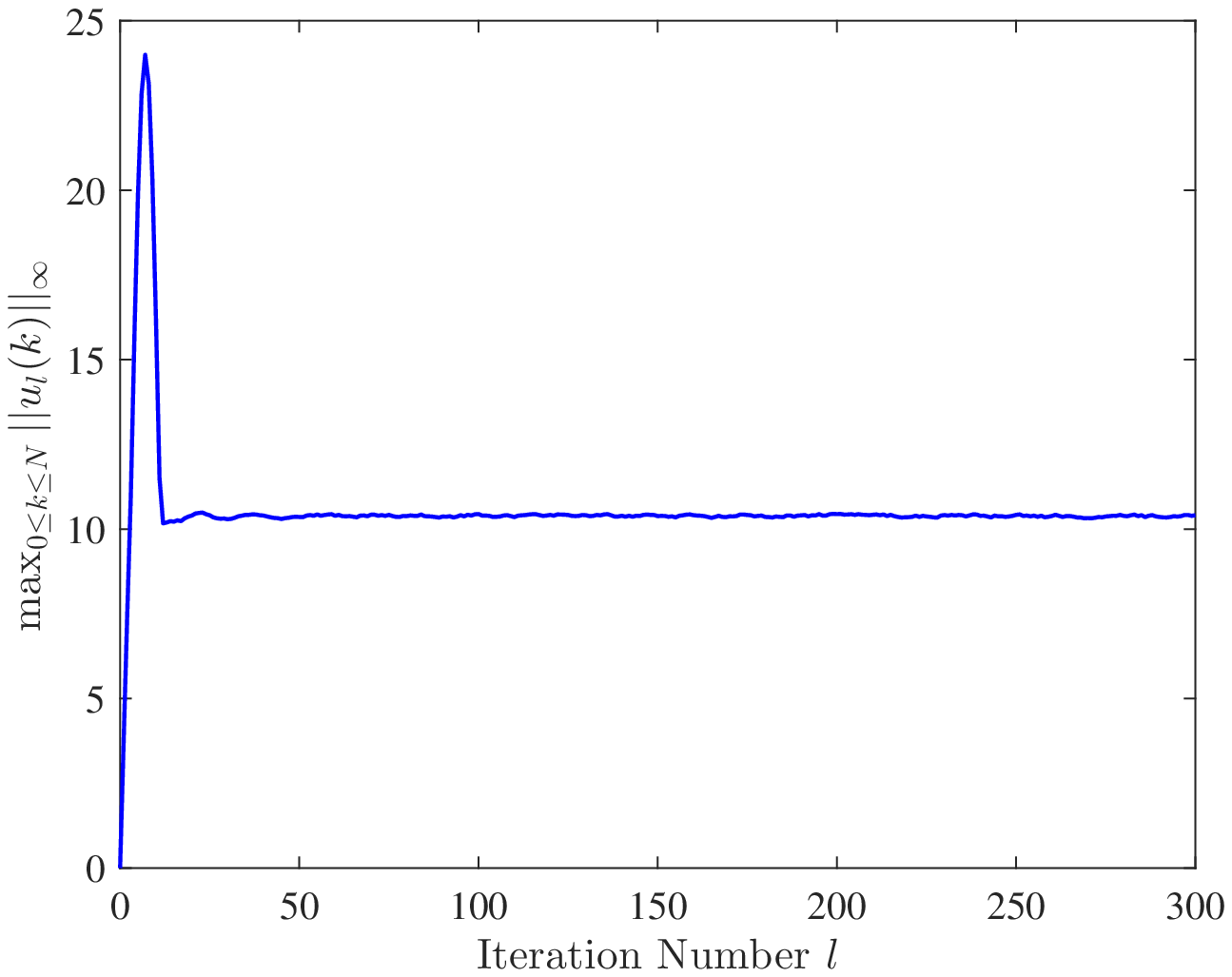}
  \caption{(Example 1). The tracking performance of the ILC system (\ref{eq01}) and (\ref{eq03}) for the first 300 iterations. Upper: ILC process evaluated by $\max_{0 \leq k \leq N} \|e_l(k)\|_\infty$. Lower: input evolution evaluated by $\max_{0 \leq k \leq N} \|u_l(k)\|_\infty$.}
  \label{fig3}
\end{figure}

\section{Illustrative Examples}\label{sec4}

We provide two examples to illustrate our robust ILC results for the systems (\ref{eq01}) and (\ref{eq28}) under the action of the updating law (\ref{eq03}), respectively, where we implement the tracking task (\ref{eq02}) for $N=100$ in the presence of a desired reference trajectory described by
\[r(k)=\begin{bmatrix}
20\left(\frac{\Ds k}{\Ds100}\right)^{2}\left(1-\frac{\Ds k}{\Ds100}\right)\\
3\sin(0.02k\pi)\end{bmatrix},\quad\forall k\in\Z_{100}.
\]

\noindent To perform simulations with the updating law (\ref{eq03}), we without loss of generality adopt the zero initial input, namely, $u_{0}(k)=0$, $\forall k\in\Z_{100}$.

{\it Example 1:} Consider the system (\ref{eq01}), for which we notice (\ref{eq04}) and simulate the nonrepetitive quantities of (\ref{eq01}) in the additive term of the nominal repetitive quantities and the nonrepetitive uncertainties. Let the nominal repetitive quantities be given by
\[
\aligned
A(k)&=\begin{bmatrix}0.16&0.0&0.0&0.0\\0.01e^{0.01k}&-0.1&-0.08&\frac{\Ds0.01}{\Ds k+2}\\0.0&0.08&0.0&0.01\cos(2k)\\-0.01k&0.0&0.0&-0.3\end{bmatrix}\\
B(k)&=\begin{bmatrix}0.5&0.0&0.0\\0.0&0.8&-0.1k\\\cos(0.1k)&0.0&0.5\\0.0&4+5\sin(3k)&3k+4\end{bmatrix}\\
C(k)&=\begin{bmatrix}2&0.0&0.1\cos(0.1(k-1))&0.0\\0.2(k-1)&2&0.0&0.1\end{bmatrix}\\
D(k)&=\begin{bmatrix}1+0.1\cos^2(0.1k) & 0.5 & 0.05\cos(0.1k) \\0.0 & 2+0.5\sin(3k) & 0.4+0.1\cos(k)\end{bmatrix}\\
w(k)&=\big[0.8\cos(0.1k),0.6\sin(0.3k),0.4\cos(0.5k),\\
&~~~~~0.2\sin(0.7k)\big]^{\tp}\\
v(k)&=\big[0.2\sin(0.4k),0.5\cos(0.6k)\big]^{\tp}\\
x_{0}&=[-1,3,-2,4]^{\tp}
\endaligned\]

\noindent For the nonrepetitive uncertainties $\delta_A(l,k)$, $\delta_B(l,k)$, $\delta_C(l,k)$, $\delta_D(l,k)$, $\delta_{x_0}(l,k)$, $\delta_w(l,k)$, $\delta_v(l,k)$, and $\delta_r(l,k)$, we consider every entry of them to vary arbitrarily on $[-0.0002,0.0002]$ with respect to both the iteration number $l$ and the time step $k$, which is simulated through the MATLAB command `rand'.

For the system (\ref{eq01}), we apply the updating law (\ref{eq03}) with $\Gamma(k)\equiv0$ and choose $\Xi(k)$ as
\begin{equation*}
\aligned
\Xi(k)
&=\begin{bmatrix}
0.25+0.1\sin\left(0.1k\right)& -0.1\\
0&0.15+0.1\cos^2\left(3k\right)\\
0&0
\end{bmatrix}.
\endaligned
\end{equation*}

\noindent The simulation tests are shown in Figs. \ref{fig3} and \ref{fig4}. It can be easily observed from Fig. \ref{fig3} that the input and tracking error are both uniformly bounded and the tracking error is finally decreased to vary within a small interval. In Fig. \ref{fig4}, the output learnt via ILC only after $l=100$ iterations can almost track the reference trajectory perfectly for all time steps (including the initial time step), in spite of the ill influences of the nonrepetitive external and internal uncertainties. This is coincident with the result of robust ILC developed in Theorem \ref{thm02}.

\begin{figure}[!t]
  \centering
  \includegraphics[width=0.8\hsize]{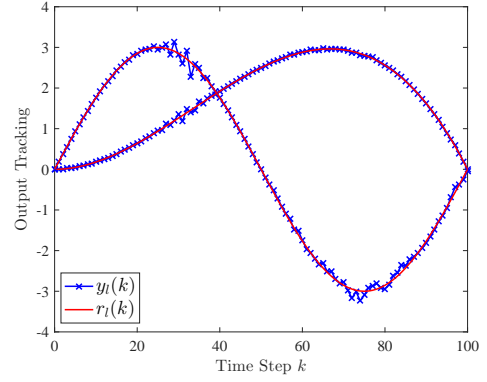}
  \caption{(Example 1). The output tracking of the ILC system (\ref{eq01}) and (\ref{eq03}) with the reference trajectory after $l=100$ iterations.}
  \label{fig4}
\end{figure}

{\it Example 2:} Consider the system (\ref{eq28}). Note that in form, (\ref{eq28}) can be viewed as a special case of (\ref{eq01}). We thus adopt $A_{l}(k)$, $B(k)$, $C(k)$, $w_{l}(k)$ and $v_{l}(k)$ in (\ref{eq28}) as the same as those of (\ref{eq01}). To apply the updating law (\ref{eq03}), we employ $\Xi(k)\equiv0$ and select $\Gamma(k)$ as
\[\Gamma(k)=\begin{bmatrix}0.3+0.1\sin(0.1k)&0\\0&0.2+0.1\cos^{2}(3k)\\0&0\end{bmatrix}.
\]

\noindent In Figs. \ref{fig5} and \ref{fig6}, we plot the simulation results. We can clearly see from Fig. \ref{fig5} that the input and tracking error are uniformly bounded and the tracking error can be decreased to vary within a small bound. In particular, Fig. \ref{fig6} shows that the output learnt via ILC can almost track the reference trajectory perfectly for all time steps (except the initial time step), regardless of the ill effects of the nonrepetitive external and internal uncertainties. This demonstrates the theoretical result of robust ILC proposed in Theorem \ref{thm04}.

\begin{figure}[!t]
\centering
\includegraphics[width=0.8\hsize]{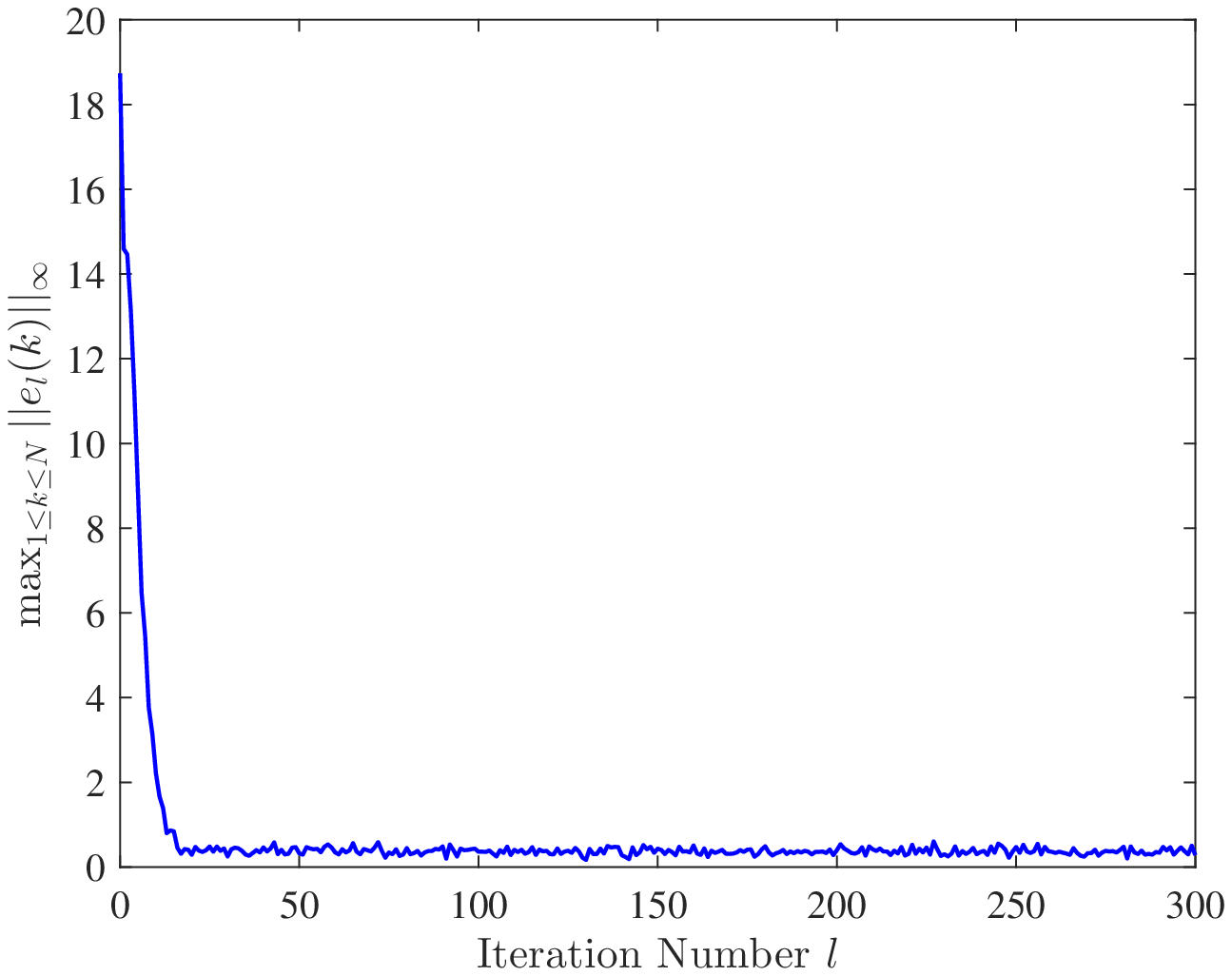}
\includegraphics[width=0.8\hsize]{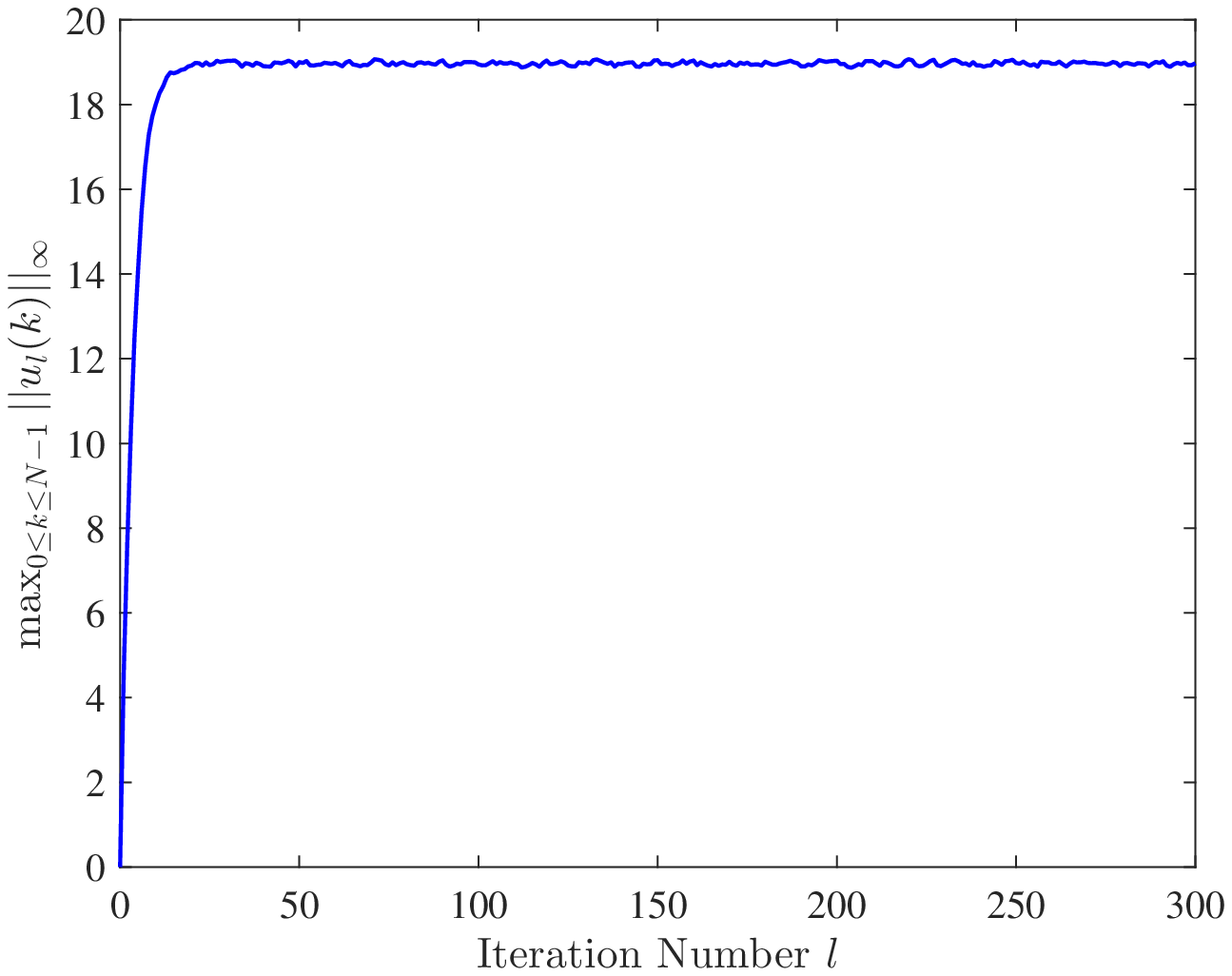}
\caption{(Example 2). The tracking performance of the ILC system (\ref{eq28}) and (\ref{eq03}) for the first $300$ iterations. Upper: ILC process evaluated by $\max_{1\leq k\leq N}\left\|e_{l}(k)\right\|_{\infty}$. Lower: input evolution evaluated by $\max_{0\leq k\leq N-1}\left\|u_{l}(k)\right\|_{\infty}$.}\label{fig5}
\end{figure}
\begin{figure}[!t]
\centering
\includegraphics[width=0.8\hsize]{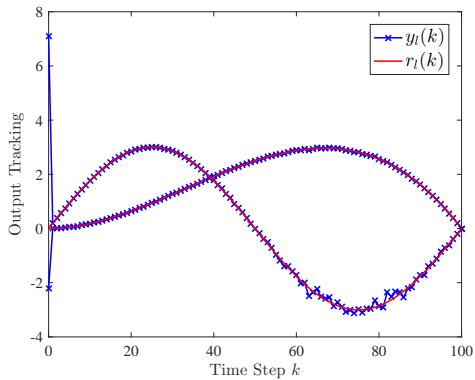}
\caption{(Example 2). The output tracking of the ILC system (\ref{eq28}) and (\ref{eq03}) with the reference trajectory after $l=100$ iterations.}\label{fig6}
\end{figure}

\section{Conclusions}\label{sec5}

In this paper, the fundamental convergence analysis problem of general MIMO ILC systems in the presence of nonrepetitive uncertainties has been studied, for which two classes of SET approaches corresponding to the ILC systems under different nonrepetitive uncertainties have been proposed. It has been validated that after the SET analysis of ILC, the updating patterns of the input can be decomposed into two essentially different classes such that the robust output tracking can be accomplished with a unified convergence condition. Moreover, it has been disclosed that to track a $p$-channel desired reference trajectory, exactly $p$ inputs are needed to be iteratively updated to learn the corresponding desired input, whereas the other input channels always remain unchanged for all iterations. This provides new insights into the convergence analysis of MIMO ILC systems and resolves the condition contradiction between the direct and indirect analysis approaches to ILC. The validity of our derived robust ILC method and results has been verified through two examples.


\section*{Acknowledgements}


The authors would like to express their thanks to Prof. K. L. Moore, Colorado School of Mines, USA, and Miss Y. Wu, Beihang University, P. R. China, for their helpful discussions.


\end{document}